\documentclass[leqno]{article}

\usepackage[letterpaper]{geometry}
\usepackage{algorithmic}
\usepackage[algo2e]{algorithm2e}

\usepackage{graphicx}

\usepackage{ltexpprt}

\usepackage{url}

\newcommand{\QQ}{{\cal A}}
\newcommand{\SSS}{{\cal S}}
\newcommand{\TT}{{\cal T}}

\newcommand{\XX}{{\cal X}}

\newcommand{\tw}{{\mathop{\rm tw}}}
\newcommand{\fillin}{{\mathop{\rm fill}}}

\newcommand{\maxatom}{{\mathop{\rm MA}}}
\mathchardef\mhyphen="2D

\begin{document}


\title{\Large A heuristic for listing almost-clique minimal separators of a
graph}
\author{Hisao Tamaki\\
Department of Computer Science, Meiji University\\
Tama, Kawasaki, Japan\\
\tt{hisao.tamaki@gmail.com}}

\date{}

\maketitle







\begin{abstract}
Bodlaender and Koster (Discrete Mathematics 2006) introduced the
notion of almost-clique separators in the context
of computing the treewidth $\tw(G)$ of a given graph $G$.
A separator $S \subseteq V(G)$ of $G$ is an
\emph{almost-clique separator} if
$S \setminus \{v\}$ is a clique of $G$ for some $v \in S$.
$S$ is a \emph{minimal separator} if $S$ has at least two
full components, where a full component of $S$ is a connected component
$C$ of $G \setminus S$ such that $N_G(C) = S$.
They observed that if $S$ is an almost-clique minimal separator of $G$
then $\tw(G \cup K(S)) = \tw(G)$, where $K(S)$ is the complete graph
on a vertex set $S$: in words, filling an almost-clique minimal separator
into a clique does not increase the treewidth.
Based on this observation, they proposed a preprocessing method for
treewidth computation, a fundamental step of which is to
find a preferably maximal set of pairwise non-crossing almost-clique minimal
separators of a graph.

In this paper, we present a heuristic for this step, which is based on the
following empirical observation. For a graph $G$ and
a minimal triangulation $H$ of $G$, let
$\QQ(H, G)$ denote the set of all almost-clique minimal separators of $G$
that are minimal separators of $H$.
Note that since the minimal separators of $H$ are pairwise
non-crossing, so are those in $\QQ(H, G)$. We observe from experiments
that $\QQ(H, G)$ is remarkably close to maximal, especially when the minimal
triangulation $H$ is computed by an algorithm aiming for small treewidth.
For example, consider the 200 graph instances from the exact
treewidth track of PACE 2017 algorithm implementation challenge. For
each instance $G$, we compute a minimal triangulation $H$ of $G$
using a variant of the MMD algorithm due to Berry {\it et al.},
extract $\QQ = \QQ(H, G)$, and then expand $\QQ$ into a maximal set
$\QQ'$ of pairwise non-crossing almost-clique minimal separators of $G$.
Then, we have $\QQ' = \QQ$ for 79 instances,
$|\QQ'| \leq 1.1 \cdot |\QQ|$ for 194 instance, and
$|\QQ'| \leq 1.24 \cdot |\QQ|$ for all the 200 instances.

This observation leads to an efficient
implementation of the preprocessing method proposed by Bodlaender and Koster.
Experiments on instances from PACE 2017 and other sources
show that this implementation is extremely fast and
effective for graphs of practical interest.
\end{abstract}

\newpage
\pagenumbering{arabic}
\pagestyle{plain}
\section{Introduction}
\label{sec:intro}
Treewidth is a graph parameter which
plays an essential role in the graph minor theory 
\cite{robertson1986graph,robertson1995graph,robertson2004graph}
and is an indispensable tool in designing graph algorithms 
(see, for example, a survey \cite{bodlaender2006treewidth}).
See Section~\ref{sec:prelim} for the definition of treewidth and
tree-decompositions. Deciding the treewidth of a given graph is NP-complete
\cite{arnborg1987complexity}, 
but admits a fixed-parameter linear time algorithm \cite{bodlaender1996linear}.

Practical algorithms for treewidth have also been actively studied
\cite{bodlaender2010treewidth,bodlaender2011treewidth,bannach2017jdrasil,tamaki2019positive,
tamaki2019computing,althaus2021ontamaki}, with recent rapid progresses 
stimulated by PACE 2016 and 2017 \cite{dell2018pace} 
algorithm implementation challenges.

The goal of our present work is to re-evaluate a classical approach to
preprocessing for treewidth computation due to Bodlaender and Koster 
\cite{bodlaender2006safe}, namely that of \emph{almost-clique separator
decomposition}, in the modern setting. This approach
was shown to be effective for relatively small graphs by experiments in their
original work but, somewhat surprisingly, no work is found in the literature
that attempts to evaluate the approach on larger graph instances that are
becoming practically tractable for treewidth computation.

Let us first review their approach. Let $\tw(G)$ denote the treewidth 
of a graph $G$. 
A separator $S \subseteq V(G)$ of $G$ is a \emph{minimal
separator} of $G$ if, for at least two connected components $C$ 
of $G[V(G) \setminus S]$, the neighborhood of $C$ is $S$.
Suppose that $G$ has a \emph{clique minimal separator} $S$, 
a minimal separator that is also a clique.
Then, $\tw(G)$ is the maximum of $\tw(G[C \cup
N_G(C)])$ where $C$ ranges over all connected components of 
$G \setminus S$. Thus, a clique minimal separator can be used to reduce
the task of computing treewidth to the tasks on separated parts. 
Given $G$, all clique minimal separators of $G$ can be listed by an
algorithm due to Tarjan \cite{tarjan1985decomposition} in $O(n m)$ time,
where $n$ and $m$ are the numbers of vertices and edges of
$G$ respectively. Indeed, his algorithm constructs 
a \emph{clique-separator decomposition} of $G$, which is
a tree-decomposition $\TT$ of $G$ satisfying the following conditions:
\begin{enumerate}
\item the intersection of every pair of adjacent bags in $\TT$ is a clique
minimal separator of $G$, and
\item for every bag $X$ of $\TT$, $G[X]$ does not contain a
clique separator.
\end{enumerate}
Let $\TT$ be a clique-separator decompositoin of $G$.
Following Tarjan, we call $G[X]$ for each bag $X$ of $\TT$ 
an \emph{atom} of the decomposition. 
Repeatedly applying the reduction described above, we see that 
$\tw(G)$ equals the maximum of $\tw(A)$, where $A$ ranges over all atoms of $\TT$.

A separator $S$ of $G$ is \emph{safe for treewidth} \cite{bodlaender2006safe},
or simply safe for short, if $\tw(G \cup K(S)) = \tw(G)$, where $K(S)$ denotes the complete graph on
vertex set $S$. If we find a minimal separator $S$ that is safe in $G$ then, 
we may treat $S$ as if $S$ is a clique minimal separator and apply the reduction
described above.
A separator $S$ of $G$ is an \emph{almost-clilque separator} of $G$ 
if there is a vertex $v \in S$ such that $S \setminus \{v\}$ is a clique of $G$.
One of the sufficient conditions for $S$ being safe, due to Bodlaener and Koster, is that
$S$ is an almost-clique minimal separator. Based on this observation, they proposed a preprocessing approach to treewidth computation, 
which we formulate as follows.

\begin{algorithm2e}
\caption{Computing treewidth with almost-clique separator based
preprocessing}
\label{alg:tw-acs}
\renewcommand\algorithmicrequire{\textbf{Input:}}
\renewcommand\algorithmicensure{\textbf{Output:}}
\begin{algorithmic}[1]
\REQUIRE{Graph $G$}
\ENSURE{$tw(G)$}
\REPEAT 
\STATE{$\QQ \Leftarrow$ a set of pairwise non-crossing
almost-clique minimal separators of $G$}
\STATE{Update $G$ by filling each $S \in \QQ$ into a clique}
\UNTIL{$G$ becomes unchagend}
\STATE{Construct a clique separator decomposition $\TT$ of $G$}
\STATE{Compute the treewidth of each atom in $\TT$}
\RETURN{the maixmum of $\tw(A)$ over all atoms $A$ of $\TT$}
\end{algorithmic} 
\end{algorithm2e}

We call the clique-separator decomposition $\TT$ of the filled graph,
constructed at line 5 of the algorithm, an \emph{almost-clique separator
decomposition} of the original graph. To compute a set of pairwise non-crossing almost-clique separators of $G$
at line 2, Bodlaender and Koster suggest applying Tarjan's algorithm
\cite{bodlaender2006safe} to list clique separators of $G \setminus \{v\}$ for every $v \in V(G)$. 
We formulate their suggestion as Algorithm~\ref{alg:list-acs}, which we
call the \emph{standard algorithm} for listing almost-clique minimal separators.
\begin{algorithm2e}
\caption{Standard algorithm for listing almost-clique minimal
separators}
\label{alg:list-acs}
\renewcommand\algorithmicrequire{\textbf{Input:}}
\renewcommand\algorithmicensure{\textbf{Output:}}
\begin{algorithmic}[1]
\REQUIRE{Graph $G$}
\ENSURE{A maximal set of pairwise non-crossing almost-clique minimal
separators of $G$}
\STATE{$\QQ \Leftarrow \emptyset$}
\FOR{each $v \in V(G)$}
\STATE{Use Targan's algorithm to compute the set $\SSS_v$ of
all clique separators of $G \setminus \{v\}$}
\FOR{each $S \in \SSS_v$}
\IF{$S \cup \{v\}$ is a minimal separator of $G$ not crossing any
separator in $\QQ$}
\STATE {add $S \cup \{v\}$ to $\QQ$}
\ENDIF
\ENDFOR
\ENDFOR
\RETURN{$\QQ$}
\end{algorithmic} 
\end{algorithm2e}

Although the running time of the standard algorithm is
polynomial, the complexity of $O(n^2 m)$ is too high to be practically
applicable to large graphs. This might be the reason why the authors of modern implementations of treewidth algorithms
\cite{bannach2017jdrasil,tamaki2019positive,althaus2021ontamaki} 
have not adopted the preprocessing method based on almost-clique
separators as proposed by Bodlaender and Koster.
We need a practically faster alternative to the standard algorithm in order to
exploit the potential of the approach.

We present a heuristic algorithm for this task, 
which is based on the following empirical observation.
For a graph $G$ and
a minimal triangulation $H$ of $G$, let $\SSS(H)$ denote the
set of minimal separators of $H$ and 
$\QQ(H, G)$ the set of almost-clique minimal separators of $G$ belonging to 
$\SSS(H)$. Note that, since the separators in $\SSS(H)$ are pairwise
non-crossing, so are those in $\QQ(H, G)$. We observe from experiments that $\QQ(H, G)$ is remarkably close to maximal, especially when the minimal
triangulation $H$ is computed by an algorithm aiming for small treewidth.

For example, consider the 200 graph instances from the exact treewidth track
of PACE 2017 algorithm implementation challenge \cite{dell2018pace}. For 
each instance $G$, we compute a minimal triangulation $H$ of $G$
using a variant, which we call MMAF, of the MMD algorithm due to Berry {\it et
al.} \cite{berry2003minimum}, extract $\QQ = \QQ(H, G)$, and then
expand $\QQ$ into a maximal set $\QQ'$ of pairwise non-crossing almost-clique minimal separators of $G$.
For this expansion, we use the list of all almost-clique minimal separators of
$G$, constructed in advance, and apply a straightforward greedy algorithm.
Then, we have $\QQ' = \QQ$ for 79 instances, $|\QQ'| \leq 1.1 \cdot |\QQ|$ for
194 instances, and $|\QQ'| \leq 1.24 \cdot |\QQ|$ for all the 200 instances.
See Section~\ref{sec:experiment} for more details and
the results when minimal triangulation algorithms other than MMAF are used.

This observation suggests that, at line 2 of Algorithm~\ref{alg:tw-acs}, 
we let $\QQ = \QQ(H, G)$ where $H$ is a minimal triangulation of $G$ computed
by the MMAF algorithm. Although $\QQ$ thus computed is not necessarily maximal,
it is expected to be close to maximal, as the above observation shows. Moreover,
except for the last round of the iteration, the non-maximality may not be a serious drawback since some of the missed
non-crossing almost-clique minimal separators may be discovered in subsequent
rounds. We experimentally compare this implementation of
Algorithm~\ref{alg:tw-acs} with the standard implementation
where Algorithm~\ref{alg:list-acs} is used to compute $\QQ$.

When applied to the 200 graph instances from the exact treewidth track of PACE 2017, 
our implementation is by orders of magnitudes faster than the standard
one. 
Let $\rho_1(G)$ denote the ratio $t_{\rm ours}(G) / t_{\rm standard}(G)$,
where $t_{\rm ours}(G)$ and $t_{\rm standard}(G)$ denote the
time spent on $G$ by our implementation and the standard implementation respectively.
See Section~\ref{sec:experiment} for the experimental environment.
Then, $\rho_1(G) \leq 0.005$ for 156 out of the 200 instances,
$\rho_1(G) \leq 0.05$ for 193 instances,
and $\rho_1(G) \leq 0.12$ for all the instances.

For the quality of decompositions, first consider the following 
measure.
Let $\maxatom_{\rm ours}(G)$ and $\maxatom_{\rm standard}(G)$ denote
the maximum numbers of vertices in an atom of the decomposition 
of $G$ produced by our implementation and by the standard
implementation, respectively, and let $\rho_2(G) = \maxatom_{\rm
ours}(G) / \maxatom_{\rm standard}(G)$.
Then, $\maxatom_{\rm ours}(G) = \maxatom_{\rm standard}(G)$ on
153 instances, 
$\rho_2(G) \leq 1.1$ for 195 instances, and 
$\rho_2(G) \leq 1.5$ for all the 200 instances.

We also consider the time for treewidth computation 
after the preprocessing, where we use our implementation of the treewidth
algorithm due to Tamaki \cite{tamaki2019computing}.
Let $t'_{\rm ours}(G)$ and $t'_{\rm standard}(G)$ denote the
time to compute the treewidth of $G$, given the almost-clique separator 
decomposition of $G$ produced by our implementation and the standard
implementation respectively, and let 
$\rho_3(G) = t'_{\rm ours}(G) / t'_{\rm standard}(G)$.
Then $\rho_3(G) \leq 1.1$ for 168 instances, 
$\rho_3(G) \leq 1.2$ for 184 instances, and
$\rho_3(G) \leq 1.5$ for all the 200 instances.

To summarize, in both measures, the performance of our
implementation is almost equal to that of the standard
implementation for most of the instances and is not
seriously inferior for any of the instances.
See Section~\ref{sec:experiment} for more details and experiments on
other benchmark instances.

The Java source code of the implementations of the algorithms used in our
experiments is available at a Github repository \url{https://github.com/twalgor/tw}.

\section{Concepts and algorithms}
\label{sec:prelim}
\setcounter{subsection}{0}
\subsection{Graph notation}

In this paper, all graphs are undirected and simple, that is, without self-loops
or parallel edges. Let $G$ be a graph. We denote by $V(G)$ the vertex set
of $G$ and by $E(G)$ the edge set of $G$. As $G$ is simple and undirected,
each member of $E(G)$ is a two-member subset of $V(G)$.
The subgraph of $G$ induced by $U \subseteq V(G)$ is denoted by 
$G[U]$. We sometimes use an abbreviation $G \setminus U$ to
stand for $G[V(G) \setminus U]$, where $U \subseteq V(G)$.
A graph $G$ is \emph{complete} if $E(V)$ contains all two-member subsets of
$V(G)$. We denote by $K(U)$ the complete graph on vertex set $U$.
For a graph $G$ and an edge set $F \subseteq E(K(V(G)))$, we deonte by
$G \cup F$ the graph $G'$ with $V(G') = V(G)$ and
$E(G') = E(G) \cup F$.

A vertex set $S \subseteq V(G)$ is a \emph{clique} of $G$ if $G[S]$ 
is a complete graph. A clique $S$ of $G$ is \emph{maximal} if no proper
superset of $S$ is a clique of $G$. In this paper, it is often needed to add
edges to $G$ so that some given vertex set $U \subseteq V(G)$ becomes a clique.
We say that we \emph{fill $U$ into a clique in $G$} in this situation;
the resulting graph is $G \cup K(U)$.
For each $v \in V(G)$, $N_G(v)$ denotes the set of neighbors of $v$ in $G$:
$N_G(v) = \{u \in V(G) \mid \{u, v\} \in E(G)\}$.
For $U \subseteq V(G)$, the {\em neighborhood of $U$ in $G$}, denoted
by $N_G(U)$, is the set of vertices adjacent to some vertex in $U$ but not
belonging to $U$ itself: $N_G(U) = (\bigcup_{v \in U} N_G(v)) \setminus U$.

We say that vertex set $C \subseteq V(G)$ is {\em connected in}
$G$ if, for every $u, v \in C$, there is a path in $G[C]$ between $u$
and $v$. It is a {\em connected component} or simply a {\em component} 
of $G$ if it is connected and is inclusion-maximal subject to this condition.
A vertex set $S \subseteq V(G)$ is a \emph{separator} of $G$
if $G \setminus S$ has more than one component. 
For brevity, we will sometimes refer to those components as \emph{the components
of $S$ in $G$} when our intention is clear 
(that we are not talking of components of $G[S]$).
A component $C$ of $S$ in $G$ is \emph{full} if $N_G(C) = S$.
A separator $S$ of $G$ \emph{separates} two vertices $a$ and $b$ if
$a$ and $b$ belong to distinct components of $S$ in $G$. We also say that
$S$ is an \emph{$a$-$b$ separator} in this situation. $S$ is a \emph{minimal
$a$-$b$ separator} if it is an $a$-$b$ separator and is inclusion-minimal
subject to this condition. $S$ is a \emph{minimal separator} if it is
a minimal $a$-$b$ separator for some pair of vertices $a$ and $b$. 
It is straightforward to see that $S$ is a minimal separator if and only if
$S$ has at least two full compoents in $G$. 
Let $R$ and $S$ be separators of $G$. We say that $R$ \emph{crosses} $S$ if
$R$ separates some pair of vertices in $S$. When $R$ and $S$ are
minimal separators, $R$ crosses $S$ if and only if $S$ crosses $R$.
\subsection{Tree-decompositions}

A \emph{tree-decomposition} of a graph $G$ is a pair $(T, \XX)$ where $T$ is a
tree and $\XX$ is a family $\{X_i\}_{i \in V(T)}$ of vertex sets of $G$, indexed
by the nodes of $T$, such that the following three conditions are satisfied. 
We call each $X_i$ the \emph{bag} at node $i$. 
\begin{enumerate}
\item $\bigcup_{i \in V(T)} X_i = V(G)$.
\item For each edge $\{u, v\} \in E(G)$, there is some $i \in V(T)$
such that $u, v \in X_i$.
\item For each $v \in V(G)$, the set of nodes 
$I_v = \{i \in V(T) \mid v \in X_i\} \subseteq V(T)$ is connected in $T$.
\end{enumerate}
The \emph{width} of this tree-decomposition is $\max_{i \in V(T)} |X_i| - 1$.
The \emph{treewidth} of $G$, denoted by $\tw(G)$ is the smallest $k$ such
that there is a tree-decomposition of $G$ of width $k$.
A tree-decomposition of $G$ is \emph{optimal} if its width equals $\tw(G)$.

The following facts are easy to verify.
\begin{proposition}
\label{prop:clique-bag}
Let $G$ be a graph and $K$ a clique of $G$.
Then, every tree-decomposition of $G$ contains a bag that is a superset of $K$.
\end{proposition}
\begin{corollary}
\label{cor:tw-clique}
For every graph $G$, $\tw(G) \geq \omega(G) - 1$ holds, where
$\omega(G)$ is the clique number of $G$.
\end{corollary}

Let $G$ be a graph and $(T, \XX)$ a tree-decomposition of $G$.
For a pair $i$, $j$ of adjacent nodes in $T$, 
let $T(i, j)$ denote the maximal subtree of $T$ containing $i$ but not $j$.
Furthermore, we define $V_{\TT}(i, j) = 
\bigcup_{t \in V(T(i, j))} X_t \setminus X_i$:
this is the union of the bags in $T(i, j)$ with vertices in $X_i$ removed.

\begin{proposition}
Let $(T, \XX)$ be a tree-decomposition of $G$. For each edge $\{i, j\}$ of
$T$, let $S_{ij}$ denote $X_i \cap X_j$. Then, $S_{ij}$ equals $V_{\TT}(i, j)
\cap V_{\TT}(j, i)$ and is an $a$-$b$ separator for every pair of vertices 
$a \in V_{\TT}(i, j) \setminus V_{\TT}(j, i)$ and 
$b \in V_{\TT}(j, i) \setminus V_{\TT}(i, j)$.
\end{proposition}
We say that each edge $\{i, j\}$ of this
tree-decomposition \emph{induces} separator $S_{ij}$.
We say that tree-decomposition $(T, \XX)$ \emph{induces} separator $S$
if some edge of $T$ induces $S$. 

\subsection{Chordal graphs and triangulations}

Tree-decompositions of graphs are closely related to 
\emph{triangulations} of graphs, defined as follows.
Let $G$ be a graph and $C$ a cycle in $G$. 
An edge $\{u, v\}$ of $G$ is a \emph{chord} of
$C$ if $u, v \in V(C)$ but $\{u, v\} \not\in E(C)$.
A graph $G$ is \emph{chordal} if every cycle $C$ of $G$
with $|V(C)| > 3$ has a chord. 
A vertex $v$ in $G$ is a \emph{simplicial vertex} of $G$
if $N_G(v)$ is a clique. A total ordering $v_1$, $v_2$, \ldots, $v_n$
of $V(G)$ is a \emph{perfect elimination order} of $G$ 
if $v_i$ is a simplicial vertex of $G[\{v_i, v_{i + 1}, \ldots, v_n\}]$
for $1 \leq i \leq n$. The following characterization of 
chordal graphs due to \cite{fulkerson1965incidence} is fundamental. 
\begin{proposition}
\label{prop:simplicial}
A graph $G$ is chordal if and only if it has a perfect elimination order.
\end{proposition}

A graph $H$ is a 
\emph{triangulation} of a graph $G$ if it is chordal,
$V(H) = V(G)$, and $E(H) \supseteq E(G)$:
it is a \emph{minimal triangulation} of $G$ if, 
furthermore, its edge set is inclusion-minimal
subject to this condition.
For a graph $G$ and a tree-decomposition $\TT = (T, \{X_i\})$ of 
$G$, let $\fillin(G, \TT)$ denote the
graph $G \cup \bigcup_{i \in V(T)} K(X_i)$ obtained by
filling every bag of $\TT$ into a clique. 

The following facts are known. (See \cite{heggernes2006minimal} for example).
\begin{proposition}
If $G$ is chordal, then there is a tree-decomposition $\TT$
of $G$ in which every bag is a maximal clique of $G$.
For every such $\TT$, the set of separators induced by edges of $\TT$ is the set of all minimal separators of $G$.
\end{proposition}

Because of this fact, minimal triangulation algorithms can be
regarded as algorithms for tree-decomposition.
Indeed, due to the following additional fact and
Corollary~\ref{cor:tw-clique}, 
$\tw(G)$ equals the smallest $k$ such that
there is a minimal triangulation with the clique number $k + 1$.
\begin{proposition}
For every graph $G$ and every tree-decomposition $\TT$ of $G$,
$\fillin(G, \TT)$ is a triangulation of $G$.
\end{proposition}

The following facts are already used in the introduction to 
reason about minimal separators of a graph $G$ obtained from a
minimal triangulation of $G$. 

\begin{proposition}
\label{prop:chordal-noncrossing}
If $G$ is chordal, then no pair of minimal separators $S_1$ and $S_2$
of $G$ cross each other.
\end{proposition}

\begin{proposition}
\label{prop:chordal-minimal}
If $H$ is a minimal triangulation of $G$, then every
minimal separator of $H$ is a minimal separator of $G$. 
\end{proposition}
\begin{corollary}
If $H$ is a minimal triangulation of $G$, then the set of all minimal
separators of $H$ is a maximal set of mutually non-crossing 
minimal separators of $G$.
\end{corollary}

\subsection{Minors}

Let $G$ be a graph and $e = \{u, v\}$ an edge of $G$. The contraction
of $e$ in $G$ is an operation to turn $G$ into a graph $G'$ in which
$u$ and $v$ are replaced by a vertex $w$ with
$N_{G'}(w) = N_G(\{u, v\})$. This vertex $w$ may be chosen to
be $u$, $v$, or any vertex not in $G$.
A graph $H$ is a \emph{minor} of $G$ if
it is obtained by a sequence of zero or more
edge contractions, vertex deletions, and edge deletions.
Let $\TT$ be a tree-decomposition of $G$. If we apply any of
these three operations to $G$ and obtain $G'$, 
$\TT$ is straightforwardly converted into a tree-decomposition
$\TT'$ of $G'$ (by replacing $u$ and/or $v$ in each bag with $w$ in case of
contracting $\{u, v\}$ into $w$) with width not larger than that of $\TT$.
Therefore, we have the following: 

\begin{proposition}
If $H$ is a minor of $G$ then $\tw(H) \leq \tw(G)$.
\end{proposition}

A minor is a \emph{clique minor} if it is a complete graph.
Let $R$ be a vertex set of $G$. A minor 
$M$ of $G$ is \emph{rooted on $R$} if $V(M) = R$ and 
each contraction in the sequence that leads from $G$ to $M$ is always
on an edge between some $u \in V(M)$ and some $v \not\in V(M)$, 
with $u$ chosen to be the vertex into which this edge is contracted.

\subsection{Clique separator decompositoins}

A separator $S$ of $G$ is a \emph{clique separator} 
if it is a clique of $G$. The following well-known fact follows immediately from 
Proposition~\ref{prop:clique-bag}.

\begin{proposition}
\label{prop:clique-sep}
Let $S$ be a cliuqe separator of $G$. Then $\tw(G)$ is
the larger of $|S| - 1$ and 
the maximum of $\tw(G[C \cup N_G(C)])$ over all components $C$ of $S$.
\end{proposition}

\begin{corollary}
\label{cor:clique-minsep}
Let $S$ be a cliuqe minimal separator of $G$. Then $\tw(G)$ is
the maximum of $\tw(G[C \cup N_G(C)])$ over all components $C$ of $S$.
\end{corollary}

A tree-decomposition $\TT$ of $G$ is a \emph{clique-separator decomosition} of
$G$ \cite{tarjan1985decomposition} if it satisfies the following conditions:
\begin{enumerate}
\item the intersection of every pair of adjacent bags in $\TT$ is a
minimal clique separator of $G$, and
\item for every bag $X$ of $\TT$, $G[X]$ does not contain a
clique separator.
\end{enumerate}
Tarjan \cite{tarjan1985decomposition} gives an $O(n m)$ time algorithm
for constructing a clique separator decomposition of a graph with $n$ vertices
and $m$ edges. Following Tarjan, we call the subgraph of $G$ induced by a bag of a clique
separator decomposition an \emph{atom} of the decomposition. 

By repeated applications of Corollary~\ref{cor:clique-minsep}, we have:

\begin{proposition}
\label{prop:clique-dec}
Let $\TT$ be a clique separator decomposition of $G$.
Then $\tw(G)$ is the maximum of $\tw(A)$ where $A$ ranges over all the
atoms of $\TT$.
\end{proposition}

\subsection{Safe separators}
\label{subsec:safe}

Bodlaender and Koster \cite{bodlaender2006safe} introduded
the notion of safe separators for treewidth.
Let $S$ be a separator of a graph $G$.
We say that $S$ is \emph{safe for treewidth}, or simply safe, 
if $\tw(G) = \tw(G \cup K(S))$.
As every tree-decomposition of $G \cup K(S)$ must have a bag containing $S$,
$\tw(G)$ is the larger of $|S| - 1$ and $\max \{\tw(G[C \cup S] \cup K(S))\}$,
where $C$ ranges over all the components of $S$.

Let us say that a separator $S$ of $G$ is \emph{minor-safe} if
for every component $C$ of $S$ in $G$, there is a clique-minor of $G \setminus
C$ rooted on $N_G(C)$. 

\begin{theorem}
\label{thm:minor-safe}
[Bodlaender and Koster\cite{bodlaender2006safe}]
If $S$ is a minor-safe separator of $G$ then $S$ is safe.
\end{theorem}

A vertex set $S \subseteq V(G)$ is an \emph{almost-clique} of $G$
if $S \setminus q$ is a clique of $G$ for some $q \in S$.
We call a separator $S$ of $G$ an \emph{almost-clique minimal separator}
if it is an almost-clique and a minimal separator at the same time.
The following observation is also due to Bodlaender and Koster
\cite{bodlaender2006safe}. 
It is originally stated for
inclusion-minimal almost-clique separators, but it is clear that
it holds more generally with almost-clique minimal separators.

\begin{proposition}
\label{prop:almost-clique-sep}
If $S$ is an almost-clique minimal separator of a graph $G$, then $S$ is
minor-safe.
\end{proposition}
\begin{proof}
Let $v \in S$ such that $S \setminus \{v\}$ is a clique.
Let $C$ be an arbitray component of $S$ in $G$. Since $S$ is a minimal
separator, $S$ has a full component $C'$ distinct from $C$.
Contracting $C'$ into $v$, we have a clique-minor of $G \setminus C$ 
rooted on $S$. Deleting vertices in $S \setminus N_G(C)$,
we have a clique-minor of $G \setminus C$ rooted on $N_G(C)$.
\end{proof}

Combining Propositions~\ref{prop:clique-dec}, \ref{prop:almost-clique-sep}
and Theorem~\ref{thm:minor-safe}, we see that
Algorithm~\ref{alg:tw-acs} for treewidth computation given 
in the introduction is correct.

\subsection{Computing minimal triangulations}
Several algorithms are known for computing a minimal triangulation of
a given graph
\cite{ohtsuki1976minimal,rose1976algorithmic,berry2004maximum,berry2003minimum},
see a survey \cite{heggernes2006minimal} for more.
In our experiments, we use MCS-M (Maximum Cardinality Search for
Minimal triangulation) 
\cite{berry2004maximum} and MMD (Minimal Minimum Degree)
\cite{berry2003minimum}.
The principal difference of these algorithms, from our perspective,
is that MMD is a good heuristic for upperbounding treewidth while MCS-M
is not intended for treewidth computation at all.

We also use a variant, we call MMAF (Minimal Minimum Average Fill), 
of MMD. To describe this variant, we need to review MMD.
MMD is based on MD 
\cite{markowitz1957elimination} (see also
\cite{bodlaender2010treewidth}), which is one of the
several heuristics for triangulation based on elimination orders.
In these heuristics, given graph $G$, a total ordering
$v_1$, $v_1$, \ldots, $v_n$ of $V(G)$ is constructed 
together with a triangulation $H$ of $G$ such that
this ordering is a perfect elimination order of $H$.
This is done as follows.
Let $H_0 = G$. At step $i$, $1 \leq i \leq n$, 
we choose the next vertex $v_i$ in the ordering and
let $H_i = (H_{i - 1} \setminus \{v_i\}) \cup K(N_{H_{i - 1}}(v_i))$.
In words, we fill the neighborhood of the chosen vertex $v_i$ in
$H_{i - 1}$, remove $v_i$, and let the resulting graph be $H_i$.
Letting $H$ be the union of $H_i$, $1 \leq i \leq n$, 
we see that the vertex ordering constructed is a perfect elimination order
of $H$. Note also that all the maximal cliques of $H$
can be found among $N_{H_{i - 1}}(v_i)$, $1 \leq i \leq n$.

In MD, $v_i$ is chosen from the vertices of the minimum degree in 
$H_{i - 1}$. In another heuristic MF (minimum fill), $v_i$ is
chosen from the vertices of minimum fill in $H_{i - 1}$, 
where the fill of $v$ in $H_{i - 1}$ is the number of
missing edges of $H_{i - 1}$ in the neighborhood of $v$.
It is observed \cite{bodlaender2010treewidth} that MF
often outperforms MD as a treewidth heuristic. In \cite{ohtsuka2018experimental},
it is observed that MAF (Minimum Average Fill) heuristic often performs 
even better, where $v_i$ is chosen from vertices $v$ such that
the fill of $v$ divided by the degree of $v$ is the smallest.

These methods based on elimination orders do not produce a minimal triangulation
in general. Berry {\it et al.} \cite{berry2003minimum} gives a
scheme of turning those methods into a minimal triangulation algorithm, 
which we sketch as follows.
Let $H$ be the triangulation of $G$ computed, say, by MD.
For each separator $S$ of $G$ that is filled into a clique in $H$,
we compute minimal separators of $G$ contained in $S$.
Rather than filling $S$ into a clique, we fill those minimal separators.
The resulting graph $G'$ is a subgraph of $H$ and, in general, is not
a triangulation of $G$. Since we have filled only minimal separators
of $G$ in $G'$, a minimal triangulation of $G'$ is a minimal triangulation
of $G$. So we apply MD to $G'$ and repeat until we get a triangulation
of $G$, which is necessarily minimal.

MMD is the result of applying the above scheme to MD. Our variant
MMAF is the result of applying the scheme to MAF. 
In Section~\ref{sec:experiment}, we will see how these two 
minimal triangulation methods together with MCS-M 
perform in our context of generating almost-clique minimal separators.

\section{Experiments}
\label{sec:experiment}
\subsection{Computational environments}
The computing environment for our experiments is as follows.
CPU: Intel Core i7-8700, 3.20GHz; RAM: 32GB; 
Operating system: Windows 10, 64bit; 
Programming language: Java 1.8; JVM: jre1.8.0\_201.
The maximum heap size is set to 28GB. The implementation is single-threaded, 
except that multiple threads may be invoked for garbage collection by JVM.
The time is measured by System.nanoTime() method and is rounded up to
the nearest millisecond.

\subsection{Graph instances}
We use two sets of instances. One is from PACE 2017 exact treewidth track 
\cite{dell2018pace} and consists of 200 instances. We call them PACE2017 instances.
The other is from the DIMACS challenge on graph-coloring \cite{johnson1996cliques}
and consists of 73 instances. We call them DIMACS instances.
Figure~\ref{fig:PACE2017} shows the PACE2017 instances: 
for each instance, a blue circle and a red point are plotted where the 
the $x$-coordinate is the number of vertices; the $y$-coordinate
is the number of edges for the blue circle and is the treewidth for
the red point. Figure~\ref{fig:DIMACS} similarly shows the
DIMACS instances. Since the exact treewidth is not known for
many of the instances in this set, the best-known upperbound on treewidth 
is used instead.

\begin{figure}[htbp]
\begin{center}
\includegraphics[width=5in, bb = 0 0 1492 847]{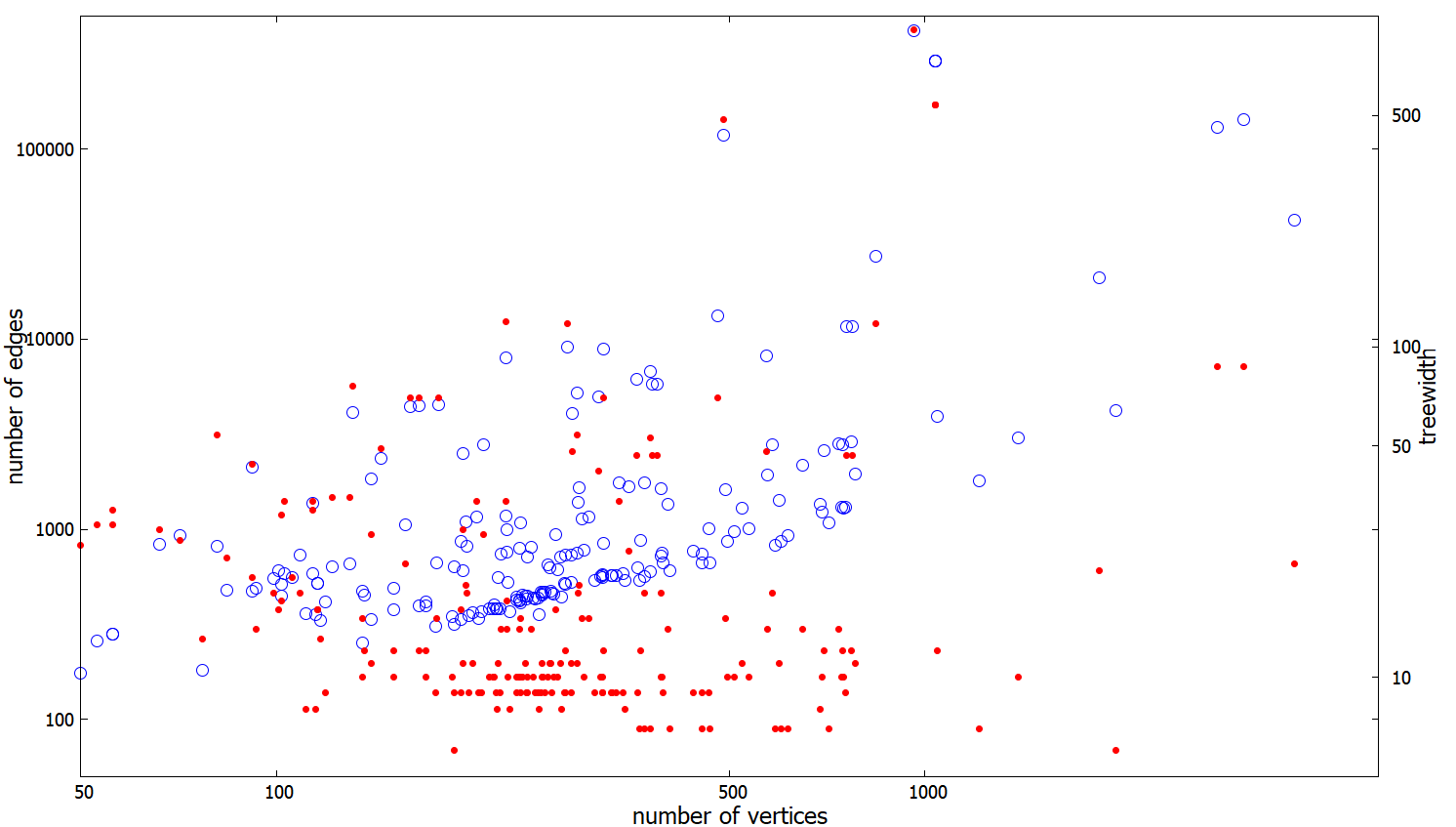}\\
\end{center}
\caption{PACE2017 graph instances}
\label{fig:PACE2017}
\end{figure} 

\begin{figure}[htbp]
\begin{center}
\includegraphics[width=5in, bb = 0 0 1492 847]{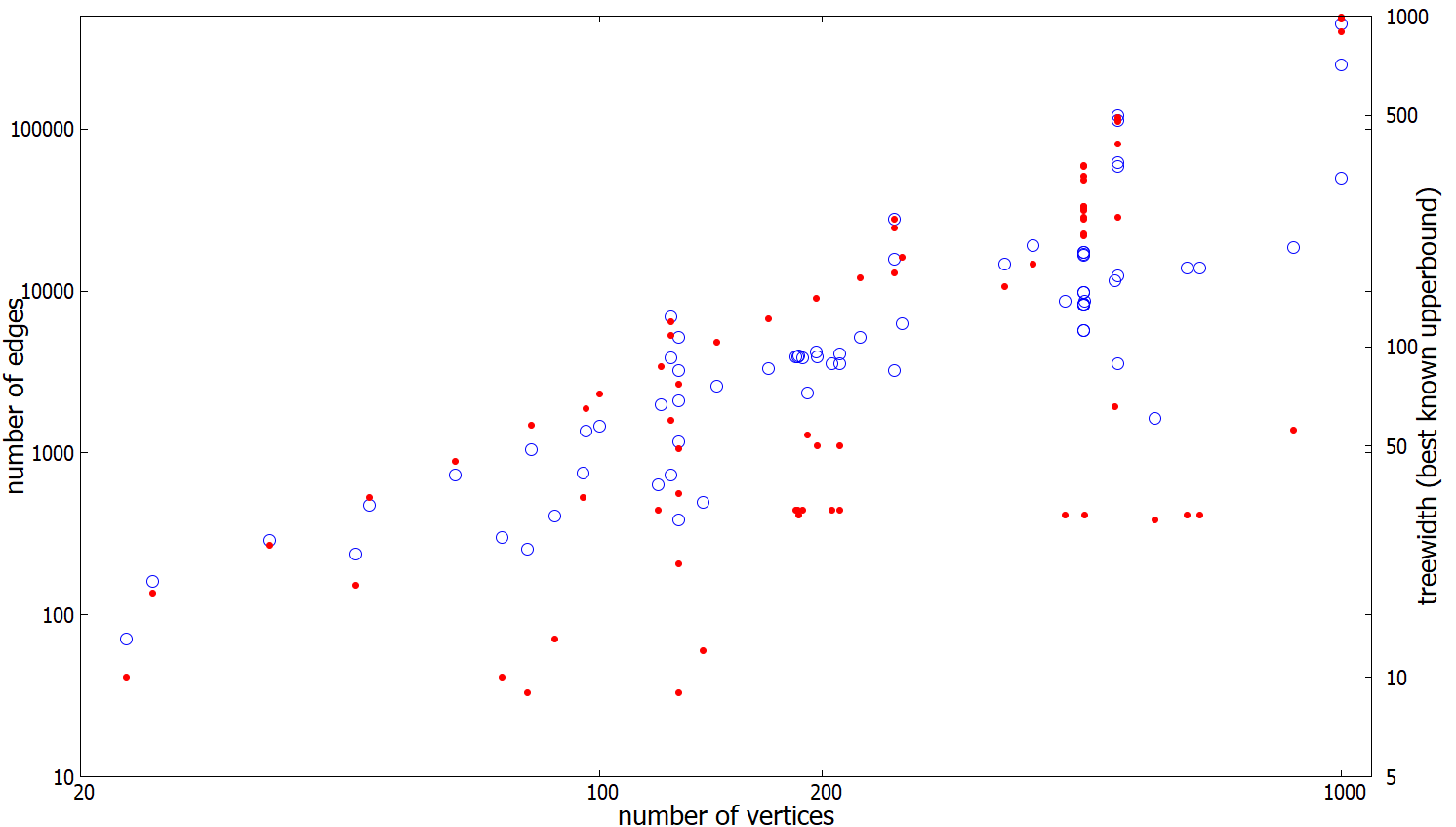}\\
\end{center}
\caption{DIMACS graph instances}
\label{fig:DIMACS}
\end{figure}

\subsection{Almost-clique separators from minimal triangultions}
\label{subsec:mininmal-triang}
Recall the notation defined in the introduction:
for graph $G$ and a minimal triangulation $H$ of $G$,
$\QQ(G, H)$ is the set of all almost-clique minimal separators
of $G$ that are minimal separators of $H$. 
In addition, we denote by $\QQ_{\rm all}(G)$ the set of all 
almost-clique minimal separators of $G$ 
and by $\QQ_{\rm max}(G)$ the maximal
subset of $\QQ_{\rm all}(G)$ consisting of pairwise non-crossing separators, 
computed by a greedy algorithm that scans the members of
$\QQ_{\rm all}(G)$ in a fixed ordering and adopts a member if it
does not cross any member previously adopted.
Similarly, we denote by $\QQ_{\rm max}(G, H)$ the maximal set
of pairwise non-crossing almost-clique minimal separators 
containing $\QQ(G, H)$ obtained by the same greedy procedure 
with the initial set $\QQ(G, H)$.

In this subsection, we experiment on the closeness
of $\QQ(G, H)$ to $\QQ_{\rm max}(G)$ and to $\QQ_{\rm max}(G, H)$ 
for each $G$ from PACE2017 instances, comparing the 
methods MMD, MMAF, MCS-M for computing the minimal triangulation $H$.

Figure~\ref{fig:acms_all}
show $|\QQ_{\rm all}(G)|$ and $|\QQ_{\rm max}(G)|$ for each $G$
in the increasing order of $|\QQ_{\rm max}(G)|$. 
There are 9 out of 200 instances for which $|\QQ_{\rm all}(G)|$ is empty:
these instances are omitted from these and subsequent figures.
For each instance, $|\QQ_{\rm max}(G)|$ is represented by the black bar and 
the difference $|\QQ_{\rm all}(G)|- |\QQ_{\rm max}(G)|$ is 
represented by the gray bar. From these figures, we see that
PACE2017 instances are abundant in almost-clique minimal separators and
the number of pairwise non-crossing ones is also large.
The median of $|\QQ_{\rm max}(G)|$ among the 200 instances is 101, achieved
by instances ex069 and ex150: for ex069, $|V(G)| = 235$, $|\QQ_{\rm max}(G)| =
100$, and $|\QQ_{\rm all}(G)|$ is 148; 
for ex150, $|V(G)| = 839$, $|\QQ_{\rm max}(G)| = 102$, 
and $|\QQ_{\rm all}(G)|$ is 161.
The maximum of $|\QQ_{\rm max}(G)|$ is 716, achieved by instance ex109
with 1212 vertices, for which $|\QQ_{\rm all}(G)|$ is 1588.

\begin{figure}[htbp]
\begin{center}
\includegraphics[width=5in, bb = 0 0 926 435]{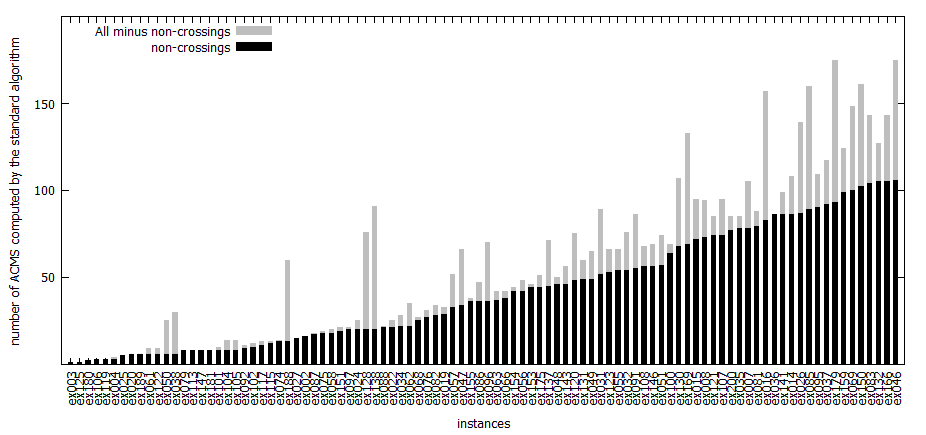}\\
\end{center}
\begin{center}
\includegraphics[width=5in, bb = 0 0 926 435]{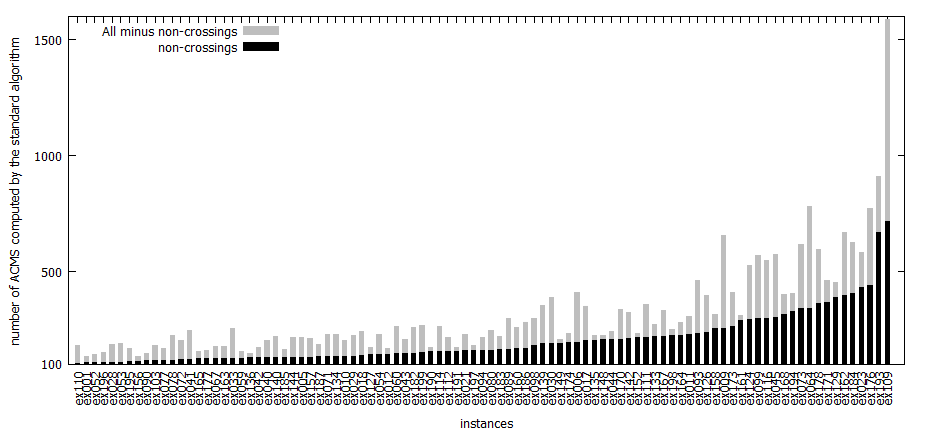}\\
\end{center}
\caption{The number of almost-clique minimal separators for each PACE2017
instance}
\label{fig:acms_all}
\end{figure} 

\begin{figure}[htbp]
\begin{center}
\includegraphics[width=5in, bb = 0 0 926 435]{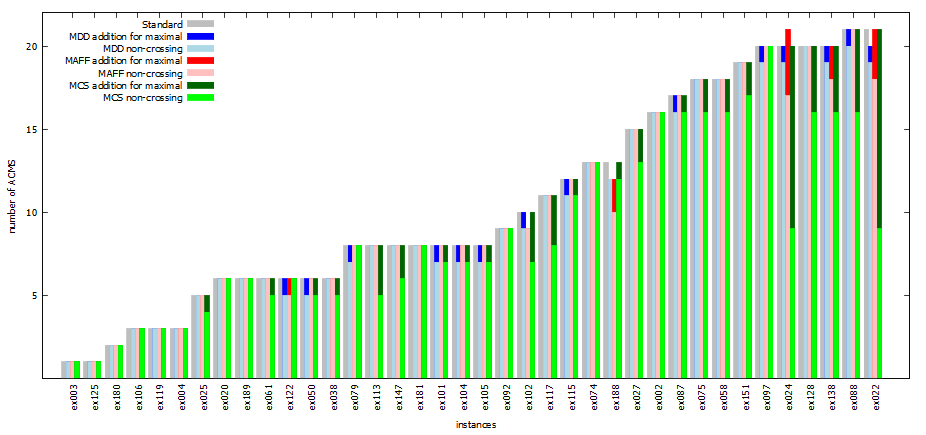}\\
\end{center}
\begin{center}
\includegraphics[width=5in, bb = 0 0 926 435]{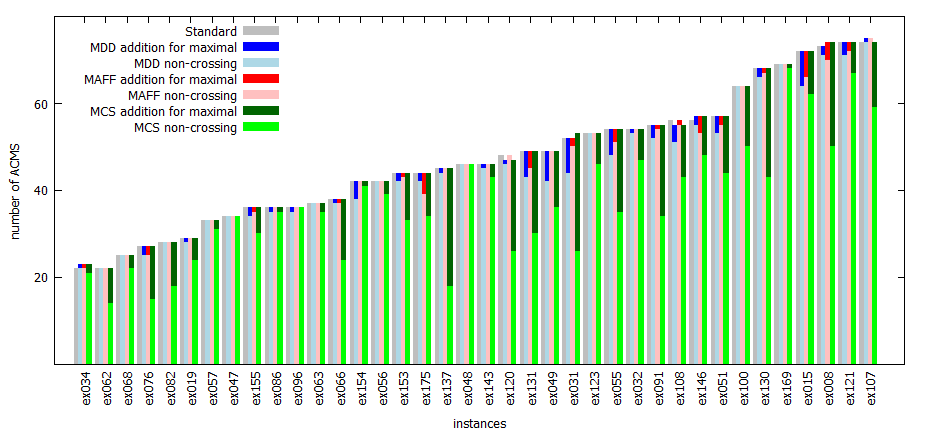}\\
\end{center}
\begin{center}
\includegraphics[width=5in, bb = 0 0 926 435]{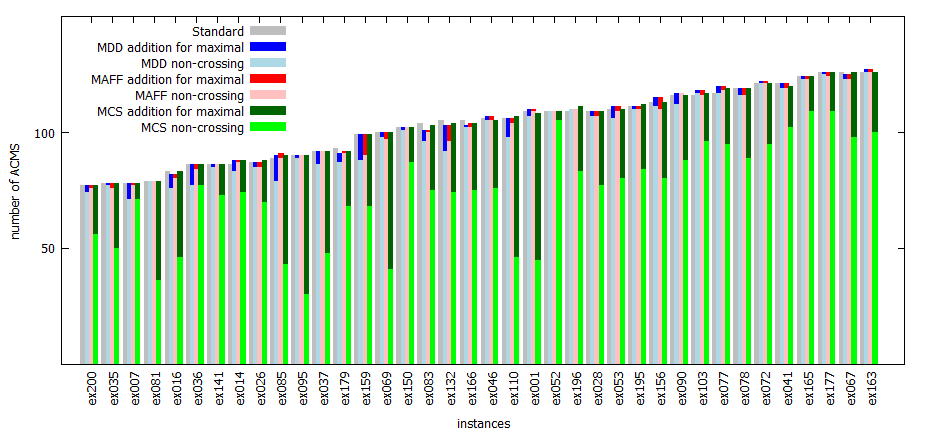}\\
\end{center}
\caption{The number of almost-clique minimal separators from
minimal triangulations for each PACE2017 instance: Groups 1, 2, 3}
\label{fig:acms_compare1}
\end{figure} 

\begin{figure}[htbp]
\begin{center}
\includegraphics[width=5in, bb = 0 0 926 435]{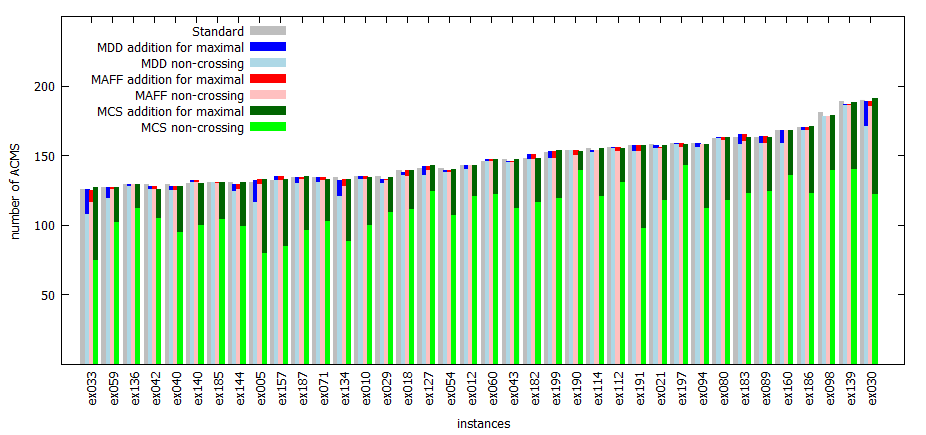}\\
\end{center}
\begin{center}
\includegraphics[width=5in, bb = 0 0 926 435]{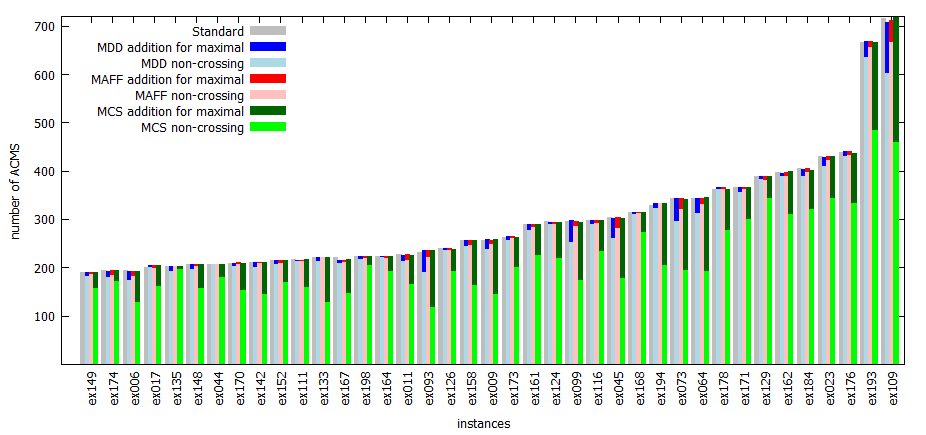}\\
\end{center}
\caption{The number of almost-clique minimal separators from
minimal triangulations for each PACE2017 instance: Groups 4, 5}
\label{fig:acms_compare2}
\end{figure}

Figures~\ref{fig:acms_compare1} and \ref{fig:acms_compare1} compare $|\QQ(G,
H)|$, for each PACE2017 instance $G$, where the minimal triangulation $H$
of $G$ is computed by three methods MMD, MMAF, and MCS-M.
The instances are divided into five subfigures, grouped in the
increasing order of $|\QQ_{\rm max(G)}|$: the first three groups are 
in Figure~\ref{fig:acms_compare1} and the remaining two are in
Figure~\ref{fig:acms_compare2}. For each instance, the gray bar represents
$|\QQ_{\rm max(G)}|$, the light-colored bars represent $|\QQ(G, H)|$, and dark-colored bars represent $|\QQ_{\rm max}(G, H) \setminus \QQ(G, H)|$, where blue is for MDD, 
red is for MMAF, and green is for MCS-M. We observe the following:
\begin{enumerate}
\item The three maximal sets $\QQ_{\rm max}(G, H)$, where
$H$ is computed by the three methods, have cardinalities
similar to each other and to the cardinality of $\QQ_{\rm max}(G)$
computed by the standard method.
\item MMD performs best, in terms of the
gap $|\QQ_{\rm max}(G, H) \setminus \QQ(G, H)|$ between the computed set 
of almost-clique minimal separators and the
maximally expanded set; MDD is slightly inferior
and MCS-M is by far inferior to the other two methods.
\end{enumerate}
The second point can be confirmed by
Table~\ref{tab:acms_compare}, where, for each method and each bound 
$\rho$ on the ratio, the number of instances satisfying
$|\QQ_{\rm max}(G, H)| / |\QQ(G, H)| \leq \rho$ is listed.
\begin{table}[htbp]
\begin{center} 
\begin{tabular}{|r|r|r|r|r|r|r|r|r|r|r|} 
\hline
&1.0&1.1&1.2&1.3&1.4&1.5&1.7&3.0\\ \hline
{\tiny MMD}&{\small 50}&{\small 173}&{\small 199}&{\small 200}&& & &
\\ \hline

{\tiny MMAF}&{\small 79}&{\small 194}&{\small 199}&{\small 200}& & & & \\
\hline 
{\tiny MCS-M}&{\small 27}&{\small 42}&{\small 78}&{\small 115}&{\small
150}&{\small 160}&{\small 180}&{\small 200}\\ \hline
\end{tabular} \end{center}
\caption{The numbers of instances with the rario $|\QQ_{\rm max}(G, H)| / 
|\QQ(G, H)|$ within the given bound}
\label{tab:acms_compare}
\end{table}

Table~\ref{tab:width_compare} compares the treewidth of the minimal
triangulation computed by these three methods. For each method and 
each bound $\rho$ on the ratio, the number of instances with
$\tw(H) / \tw(G) \leq \rho$, where $H$ is the minimal triangulation computed
by the method, is listed. MMAF performs the best and MCS-M is considerably
inferior, which is not surprising since MCS-M is not intended for small
treewidth.
The correlations of the performances shown by these two tables are not
surprising either. For each almost-clique minimal separator $S$ of $G$, 
there exists a tree-decomposition of $tw(G)$ that induces $S$ as a separator.
It is natural to expect that the chances of a tree-decomposition $\TT$ of $G$
inducing $S$ are larger when the width of $\TT$ is closer to $\tw(G)$.

\begin{table}[htbp]
\begin{center} 
\begin{tabular}{|r|r|r|r|r|r|r|r|r|r|r|} 
\hline
&1.0&1.1&1.2&1.3&1.4&2.6&3.9&8.8\\ \hline
{\tiny MMD} &{\small 22}&{\small 40 }&{\small 76 }&{\small 125}&{\small 
146}&{\small 200}& & \\ \hline 
{\tiny MMAF}&{\small 48}&{\small 94 }&{\small 160}&{\small 181}&{\small
194}&{\small 200}& & \\ \hline 
{\tiny MCS-M} &{\small 14}&{\small 18 }&{\small 29
}&{\small 42 }&{\small 53 }&{\small 164}&{\small 190}&{\small 200}\\ \hline
\end{tabular} \end{center}
\caption{The numbers of instances with the rario $tw(H) / tw(G)$ 
within the given bound}
\label{tab:width_compare}
\end{table}

\subsection{Almost-clique separator decomposition}
In this subsection, we compare two implementations of
Algorithm~\ref{alg:tw-acs} given in the introduction for
computing treewidth using an almost-clique separator decomposition: one uses
the standard algorithm for listing almost-clique minimal separators
and the other uses our heuristic based on minimal triangulations.
For the latter implementation, we use MMAF for computing the minimal
triangulation, as it is the best among the three methods compared
by the experiments in the previous subsection.
In the figures in this subsection, instances are plotted with circles
whose areas are proportional to the number of vertices of the instances.
The constant of proportionality, however, may not be consistent across figures.

\subsubsection{PACE2017 instances}
We start with comparisons on PACE2017 instances.

We first compare the time for computing the almost-clique separator
decomposition by the two implementations.
Recall the notation in the introduction: $t_{\rm ours}(G)$ and 
$t_{\rm standard}(G)$ denote time spent on $G$ by our 
heuristic implementation and by the standard
implementation, respectively.
Figure~\ref{fig:ppTime} plots
PACE2017 instances, where the $x$-coordinate is $t_{\rm ours}(G)$ 
and the $y$-coordinate is $t_{\rm standard}(G)$.
We see that our heuristic implementation is by orders of magnitudes faster
than the standard implementation.

\begin{figure}[htbp]
\begin{center}
\includegraphics[width=5in, bb = 0 0 589 552]{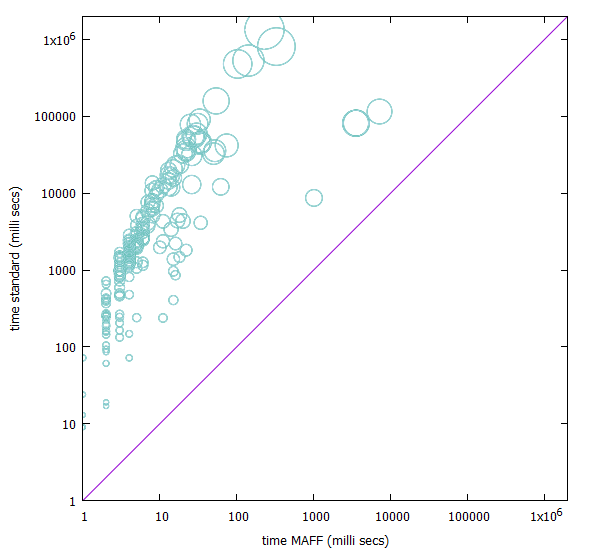}\\
\end{center}
\caption{Time for computing the almost-clique separator
decomposition of PACE2017 instances}
\label{fig:ppTime}
\end{figure} 

To compare the quality of the almost-clique separator decomposition
produced by the two implementations, we first consider the
the number of vertices in the largest atom of the decomposition.
Recall the notation in the introduction: $\maxatom_{\rm ours}(G)$ and 
$\maxatom_{\rm standard}(G)$ denote 
the maximum number of vertices in an atom of the decomposition 
of $G$ produced by our heuristic implementation and by the standard
implementation, respectively.
Figure~\ref{fig:maxAtom} plots
PACE2017 instances, where the $x$-coordinate is $|\maxatom_{\rm ours}(G)| /
|V(G)|$ and the $y$-coordinate is $|\maxatom_{\rm standard}(G)|/ |V(G)|$.
We see that the performances of the two implementations in this measure are
almost equal for most of the instances. Note also that, for a majority of
the instances, the number of vertices in the largest atom is smaller than half of
the total number of vertices: the almost-clique separator decomposition approach
itself is quite effective.
\begin{figure}[htbp]
\begin{center}
\includegraphics[width=5in, bb = 0 0 624 616]{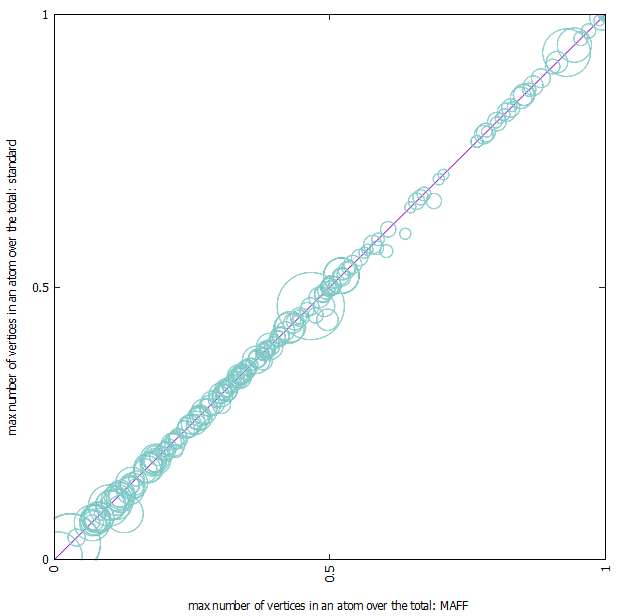}\\
\end{center}
\caption{Relative size of the largest atom in 
the almost-clique separator decomposition of PACE2017 instances}
\label{fig:maxAtom}
\end{figure} 

We next consider the time for computing $\tw(G)$ given 
the almost-clique separator decomposition.
Recall the notation in the introduction: $t'_{\rm ours}(G)$ and 
$t'_{\rm standard}(G)$ denote the time for
computing $\tw(G)$ given 
the almost-clique separator decompositionn 
of $G$ produced by our heuristic implementation and by the standard
implementation, respectively. 
For this treewidth computation, we use our implementation of the treewidth
algorithm due to Tamaki \cite{tamaki2019computing}.
Figure~\ref{fig:solveTime} plots
PACE2017 instances, where the $x$-coordinate is $t'_{\rm ours}(G)$ 
and the $y$-coordinate is $t'_{\rm standard}(G)$.
We see that the performances of the two implementations in this measure are
also almost equal for most of the instances.

\begin{figure}[htbp]
\begin{center}
\includegraphics[width=5in, bb = 0 0 589 552]{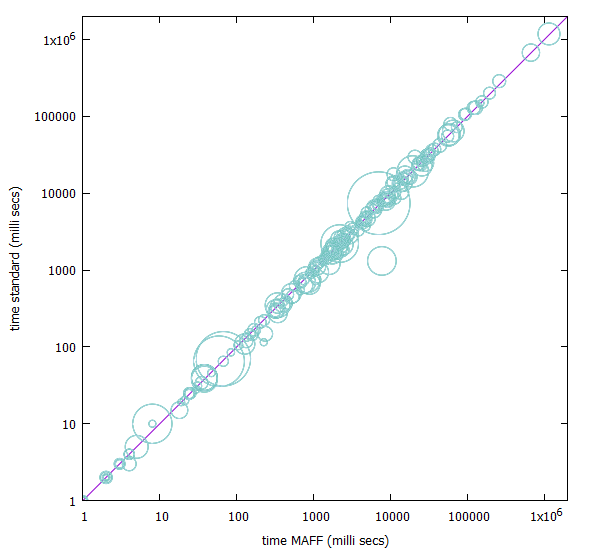}\\
\end{center}
\caption{Time for computing the treewidth of PACE2017 instances given the
almost-clique separator decomposition}
\label{fig:solveTime}
\end{figure} 

We finally compare the total time for computing the treewidth by
the two implementations.
Figure~\ref{fig:totalTime} plots
PACE2017 instances, where the $x$-coordinate is $t_{\rm ours}(G) + t'_{\rm
ours}(G)$ and the $y$-coordinate is $t_{\rm standard}(G) + t'_{\rm
standard}(G)$. We see that our heuristic implementation
consistently ourperforms the standard implementation in this measure.

\begin{figure}[htbp]
\begin{center}
\includegraphics[width=5in, bb = 0 0 589 552]{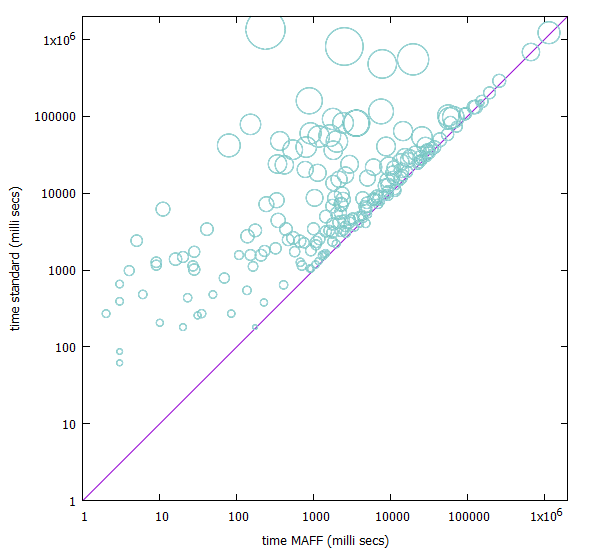}\\
\end{center}
\caption{Total time for computing the treewidth of PACE2017 instances}
\label{fig:totalTime}
\end{figure} 

\subsubsection{DIMACS graph coloring instances}
We have also compared the two implementations on DIMACS instances.
Of the total of 73 instances in this set, 41 instances have no almost-clique separators. We call them \emph{sterile} instances and
treat them separately in the following figures.

We first compare the time for computing the almost-clique separator
by the two implementations. We continue to use the
same notation: $t_{\rm ours}(G)$ and 
$t_{\rm standard}(G)$ denote time spent on $G$ by our 
heuristic implementation and by the standard
implementation, respectively.
We have two figures: Figure~\ref{fig:ppTimeColoring0} for 
sterile instances and Figure~\ref{fig:ppTimeColoring1} for non-sterile
instances.
In both figures, each instance is plotted where the $x$-coordinate 
is $t_{\rm ours}(G)$ and $y$-coordinate is$t_{\rm standard}(G)$. From these figures, we see that
the advantage of our heuristic implementation over the standard one is larger
on non-sterile instances than on sterile instances. Our implementation, 
however, is at least ten times faster than the standard one even on
sterile instances.

\begin{figure}[htbp]
\begin{center}
\includegraphics[width=5in, bb = 0 0 624
616]{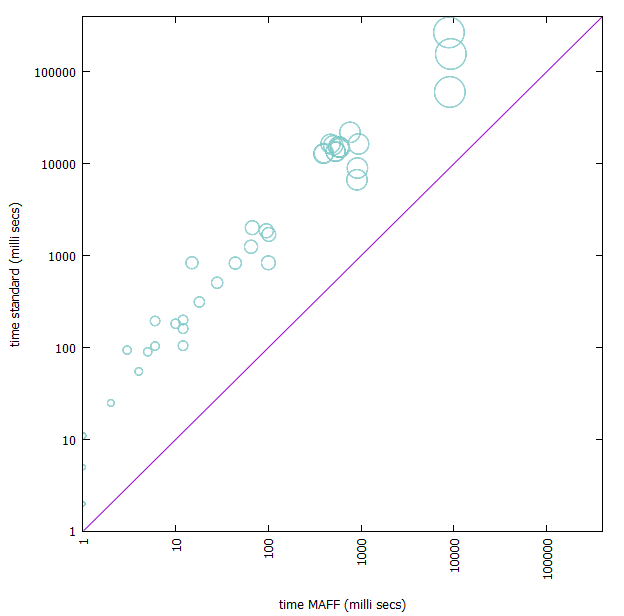}\\
\end{center}
\caption{Time for computing the almost-clique separator
decomposition of for sterile DIMACS instance}
\label{fig:ppTimeColoring0}
\end{figure} 

\begin{figure}[htbp]
\begin{center}
\includegraphics[width=5in, bb = 0 0 589
552]{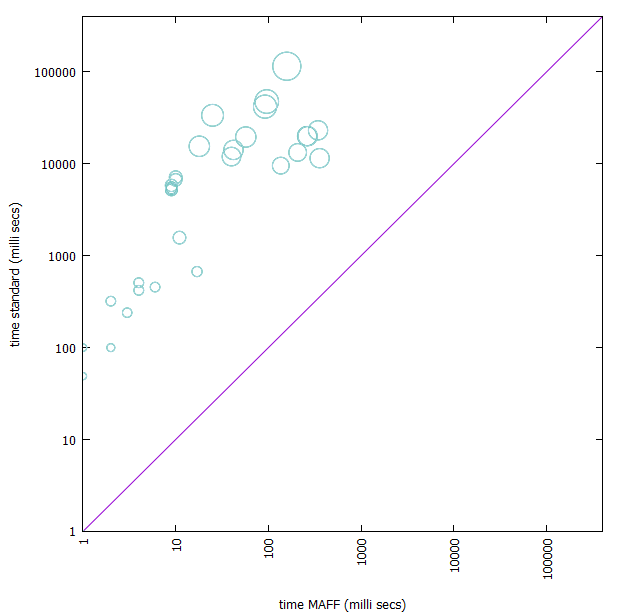}\\
\end{center}
\caption{Time for computing the almost-clique separator
decomposition of non-sterile DIMACS instances} 
\label{fig:ppTimeColoring1}
\end{figure}

We compare the size of the largest atom in the almost-clique separator
decomposition only on non-sterile instances.
In Figure~\ref{fig:maxAtomColoring}, 
each instance is plotted with $|\maxatom_{\rm
ours}(G)| / |V(G)|$ as the $x$-coordinate and 
$|\maxatom_{\rm standard}(G)|/ |V(G)|$ as the $y$-coordinate.
We see that, similarly to the case of PACE2017 instances, the performances of
the two implementations in this measure are almost equal for most of the instances. 

We do not include the comparison on the time to compute treewidth given the
almost-clique separator decompositions, because not all of the instances
in this set are solvable in a reasonable amount of time.
The comparisons of the largest atom size alone, however, show that
the quality of the almost-clique separator decompositions produced by the two
implementations are almost equal. 

The effectiveness of the almost-clique separator based 
preprocessing itself is somewhat limited on DIMACS instances 
as more than half of them are sterile. There are, however,
a non-negligible number of instances in this set where the reduction of the
problem size by this approach is dramatic.
In view of the small running time of our heuristic algorithm, 
it seems a good strategy to try this approach in general, unless
there is a strong evidence that the instance at hand is sterile.
\begin{figure}[htbp]
\begin{center}
\includegraphics[width=5in, bb = 0 0 589
552]{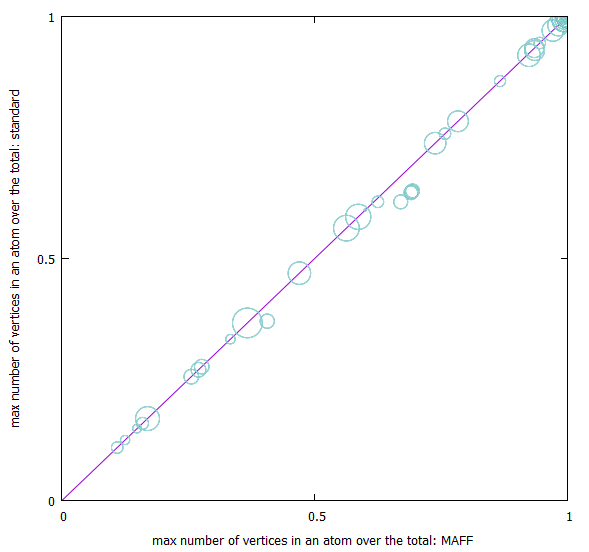}\\
\end{center}
\caption{Relative size of the largest atom in 
the almost-clique separator decomposition of DIMACS graph coloring
instances}
\label{fig:maxAtomColoring}
\end{figure} 

\section{Conclusions}
\label{sec:conc}
We have developed a practically efficient heuristic method of
listing almost-clique minimal separators of a given graph.
Because of this new method, we may now regard the preprocessing
method of Bodlaender and Koster based on almost-clique separator
decompositions as a standard component to be included in any
practical implementations of treewidth algorithms.

As stated in Subsection~\ref{subsec:safe}, this preprocessing
approach is a special case of their approach of using minor-safety as a sufficient 
condition for the safety of separators. 
The problem of deciding if a given separator is minor-safe is NP-complete.
Therefore, we need a good heuristic for this task. Although some heuristic
methods have been developed and successfully used in previous work
\cite{tamaki2019positive, althaus2021ontamaki}, the effect of applying
those heuristics is not sufficiently predictable. 
Although some hard instances become easily solvable
due to the problem reduction by the preprocessing, there are some 
other instances on which the heuristics discover no safe separators after 
expensive combinatorial searches, contributing only to a huge overhead.
It appears difficult to know in advance which would happen.
To turn this more general preprocessing approach
into a preprocessing component as stable as the one developed here, more
research is required.

\bibliography{main}

\begin{thebibliography}{10}

\bibitem{althaus2021ontamaki}
Ernst Althaus, Daniela Schnurbusch, Julian Wueschner, and Sarah Ziegler.
\newblock On tamaki's algorithm to compute treewidths.
\newblock In {\em 19th Symposium on Experimental Algorithms}, 2021.
\newblock to appear.

\bibitem{arnborg1987complexity}
Stefan Arnborg, Derek~G Corneil, and Andrzej Proskurowski.
\newblock Complexity of finding embeddings in ak-tree.
\newblock {\em SIAM Journal on Algebraic Discrete Methods}, 8(2):277--284,
  1987.

\bibitem{bannach2017jdrasil}
Max Bannach, Sebastian Berndt, and Thorsten Ehlers.
\newblock Jdrasil: A modular library for computing tree decompositions.
\newblock In {\em 16th International Symposium on Experimental Algorithms (SEA
  2017)}. Schloss Dagstuhl-Leibniz-Zentrum fuer Informatik, 2017.

\bibitem{berry2004maximum}
Anne Berry, Jean~RS Blair, Pinar Heggernes, and Barry~W Peyton.
\newblock Maximum cardinality search for computing minimal triangulations of
  graphs.
\newblock {\em Algorithmica}, 39(4):287--298, 2004.

\bibitem{berry2003minimum}
Anne Berry, Pinar Heggernes, and Genevieve Simonet.
\newblock The minimum degree heuristic and the minimal triangulation process.
\newblock In {\em International Workshop on Graph-Theoretic Concepts in
  Computer Science}, pages 58--70. Springer, 2003.

\bibitem{bodlaender1996linear}
Hans~L Bodlaender.
\newblock A linear-time algorithm for finding tree-decompositions of small
  treewidth.
\newblock {\em SIAM Journal on computing}, 25(6):1305--1317, 1996.

\bibitem{bodlaender2006treewidth}
Hans~L Bodlaender.
\newblock Treewidth: characterizations, applications, and computations.
\newblock In {\em International Workshop on Graph-Theoretic Concepts in
  Computer Science}, pages 1--14. Springer, 2006.

\bibitem{bodlaender2006safe}
Hans~L Bodlaender and Arie~MCA Koster.
\newblock Safe separators for treewidth.
\newblock {\em Discrete Mathematics}, 306(3):337--350, 2006.

\bibitem{bodlaender2010treewidth}
Hans~L Bodlaender and Arie~MCA Koster.
\newblock Treewidth computations i. upper bounds.
\newblock {\em Information and Computation}, 208(3):259--275, 2010.

\bibitem{bodlaender2011treewidth}
Hans~L Bodlaender and Arie~MCA Koster.
\newblock Treewidth computations ii. lower bounds.
\newblock {\em Information and Computation}, 209(7):1103--1119, 2011.

\bibitem{dell2018pace}
Holger Dell, Christian Komusiewicz, Nimrod Talmon, and Mathias Weller.
\newblock The pace 2017 parameterized algorithms and computational experiments
  challenge: The second iteration.
\newblock In {\em 12th International Symposium on Parameterized and Exact
  Computation (IPEC 2017)}. Schloss Dagstuhl-Leibniz-Zentrum fuer Informatik,
  2018.

\bibitem{fulkerson1965incidence}
Delbert Fulkerson and Oliver Gross.
\newblock Incidence matrices and interval graphs.
\newblock {\em Pacific journal of mathematics}, 15(3):835--855, 1965.

\bibitem{heggernes2006minimal}
Pinar Heggernes.
\newblock Minimal triangulations of graphs: A survey.
\newblock {\em Discrete Mathematics}, 306(3):297--317, 2006.

\bibitem{johnson1996cliques}
David~S Johnson and Michael~A Trick.
\newblock {\em Cliques, coloring, and satisfiability: second DIMACS
  implementation challenge, October 11-13, 1993}, volume~26.
\newblock American Mathematical Soc., 1996.

\bibitem{markowitz1957elimination}
Harry~M Markowitz.
\newblock The elimination form of the inverse and its application to linear
  programming.
\newblock {\em Management Science}, 3(3):255--269, 1957.

\bibitem{ohtsuki1976minimal}
Tatsuo Ohtsuki, Lap~Kit Cheung, and Toshio Fujisawa.
\newblock Minimal triangulation of a graph and optimal pivoting order in a
  sparse matrix.
\newblock {\em Journal of Mathematical Analysis and Applications},
  54(3):622--633, 1976.

\bibitem{ohtsuka2018experimental}
Hiromu Otsuka, Tomoki Kuida, Takumi Sato, and Hisao Tamaki.
\newblock Experimental evaluation of greedy treewidth heuristics on huge
  graphs.
\newblock In {\em SIGAL-166-12}. Information Processing Society of Japan, 2018.
\newblock in Japanese.

\bibitem{robertson1986graph}
Neil Robertson and Paul~D. Seymour.
\newblock Graph minors. ii. algorithmic aspects of tree-width.
\newblock {\em Journal of algorithms}, 7(3):309--322, 1986.

\bibitem{robertson1995graph}
Neil Robertson and Paul~D Seymour.
\newblock Graph minors. xiii. the disjoint paths problem.
\newblock {\em Journal of combinatorial theory, Series B}, 63(1):65--110, 1995.

\bibitem{robertson2004graph}
Neil Robertson and Paul~D Seymour.
\newblock Graph minors. xx. wagner's conjecture.
\newblock {\em Journal of Combinatorial Theory, Series B}, 92(2):325--357,
  2004.

\bibitem{rose1976algorithmic}
Donald~J Rose, R~Endre Tarjan, and George~S Lueker.
\newblock Algorithmic aspects of vertex elimination on graphs.
\newblock {\em SIAM Journal on computing}, 5(2):266--283, 1976.

\bibitem{tamaki2019computing}
Hisao Tamaki.
\newblock Computing treewidth via exact and heuristic lists of minimal
  separators.
\newblock In {\em International Symposium on Experimental Algorithms}, pages
  219--236. Springer, 2019.

\bibitem{tamaki2019positive}
Hisao Tamaki.
\newblock Positive-instance driven dynamic programming for treewidth.
\newblock {\em Journal of Combinatorial Optimization}, 37(4):1283--1311, 2019.

\bibitem{tarjan1985decomposition}
Robert~E Tarjan.
\newblock Decomposition by clique separators.
\newblock {\em Discrete mathematics}, 55(2):221--232, 1985.

\end{thebibliography}

\end{document}